\documentclass[aps,reprint,superscriptaddress,nofootinbib]{revtex4-1}

\pdfoutput=1
\usepackage{silence}
\usepackage[hidelinks]{hyperref}

\usepackage{natbib}
\usepackage{graphicx}
\usepackage{amsmath, amsfonts, amssymb, mathtools}
\usepackage{bm}
\usepackage{lineno}
\usepackage[dvipsnames]{xcolor}
\usepackage[T1]{fontenc}
\usepackage{comment}
\usepackage{xcolor}
\usepackage{soul}
\graphicspath{{figs/}}

\usepackage{titlesec}
\usepackage{lipsum}  
\usepackage{booktabs}

\usepackage{array}
\usepackage{wrapfig}
\usepackage{multirow}
\usepackage{tabu}
\usepackage{makecell}
\usepackage{xspace}
\usepackage{placeins}

\renewcommand{\thetable}{\arabic{table}}

\newcommand\E[2][]{\ensuremath{\mathbb{E}_{#1}\left[#2\right]\xspace}}

\usepackage{bm, amsfonts, amsmath, amssymb, xspace, mathtools, amsthm}
\usepackage{graphicx}
\usepackage{hyperref}
\usepackage[dvipsnames]{xcolor}
\newcommand\myshade{85}
\colorlet{mylinkcolor}{BrickRed}
\colorlet{mycitecolor}{NavyBlue}
\colorlet{myurlcolor}{Aquamarine}
\hypersetup{
  linkcolor  = mylinkcolor!\myshade!black,
  citecolor  = mycitecolor!\myshade!black,
  urlcolor   = myurlcolor!\myshade!black,
  colorlinks = true,
}

\newcommand\subref[2]{\hyperref[#1]{\ref*{#1}#2}}

\newtheorem{theorem}{Theorem}
\newtheorem{lemma}{Lemma}

\begin{document}

\title{Quantifying synergy and redundancy in multiplex networks}

\author{Andrea I Luppi}
\affiliation{Division of Anaesthesia and Department of Clinical Neurosciences, University of Cambridge, Cambridge, UK}
\affiliation{Montreal Neurological Institute, McGill University, Montre\'al, Canada}
\thanks{Correspondence: \url{al857@cam.ac.uk}}

\author{Eckehard Olbrich}
\affiliation{Max Planck Institute for Mathematics in the Sciences, Leipzig, Germany}

\author{Conor Finn}
\affiliation{Max Planck Institute for Mathematics in the Sciences, Leipzig, Germany}

\author{Laura E. Su\'arez}
\affiliation{Montreal Neurological Institute, McGill University, Montre\'al, Canada}

\author{Fernando E. Rosas}
\affiliation{Department of Informatics, University of Sussex, Brighton, UK}
\affiliation{Centre for complexity science, Imperial College London, London, UK}
\affiliation{Centre for psychedelic research, Department of Brain Sciences, Imperial College London, London, UK}
\affiliation{Centre for Eudaimonia and human flourishing, University of Oxford, Oxford, UK}

\author{Pedro A.M. Mediano}
\affiliation{Department of Computing, Imperial College London, London, UK}

\author{J{\"u}rgen Jost}
\affiliation{Max Planck Institute for Mathematics in the Sciences, Leipzig, Germany}
\affiliation{ScaDS.AI, Leipzig University, Germany}
\affiliation{Santa Fe Institute, Santa Fe, USA}

\hspace{.1cm}
\begin{abstract}
\section*{Abstract}
\vspace{-.2cm}
\noindent 
Understanding how different networks relate to each other is key for obtaining a greater insight into complex systems. Here, we introduce an intuitive yet powerful framework to characterise the relationship between two networks comprising the same nodes. We showcase our framework by decomposing the shortest paths between nodes as being contributed uniquely by one or the other source network, or redundantly by either, or synergistically by the two together. Our approach takes into account the networks' full topology, and it also provides insights at multiple levels of resolution: from global statistics, to individual paths of different length. We show that this approach is widely applicable, from brains to the London public transport system. In humans and across 123 other mammalian species, we demonstrate that reliance on unique contributions by long-range white matter fibers is a conserved feature of mammalian structural brain networks. Across species, we also find that efficient communication relies on significantly greater synergy between long-range and short-range fibers than expected by chance, and significantly less redundancy. Our framework may find applications to help decide how to trade-off different desiderata when designing network systems, or to evaluate their relative presence in existing systems, whether biological or artificial.

\end{abstract}
\maketitle

\section*{Introduction}

Networks provide an intuitive representation of systems comprising interacting components, enabling the use of rigorous mathematical tools to study complex systems across domains of science. A useful way of leveraging network representations involves assessing the (dis)similarity between the topologies of two networks built upon the same underlying nodes. For example, one may want to compare the train and flight transportation networks linking a given set of cities, compare different modes of social interaction between the same group of people, or contrast different types of connections bridging different brain areas. 
Approaches to this question are usually based on comparing different network topologies in terms of their `distance.' In effect, various scalar metrics to capture the divergence between different networks have been introduced for this purpose, including graph edit distance~\citep{gao2010graphedit}, graph kernel methods~\citep{borgwardt2005shortest}, spectral methods based on the graph's eigenspectrum~\citep{delange2014frontiers, suarez2022connectomics, dedomenico2015natcomm}, information-theoretic metrics based on various graph properties~\citep{dedomenico2016prx, bagrow2019information}, among others~\citep{kawahara2017brainnetcnn, mheich2017siminet, shimada2016scirep, schieber2017natcomm, cao2013levenshtein, koutra2013deltacon, lacasa2021beyond} 
 (reviewed in ~\citep{wills2020metrics,  tantardini2019comparing, mheich2020brain}). However, these approaches tend to characterise the divergence between networks in terms of a single number, thus yielding a one-dimensional description of a potentially rich discrepancy. 
 
Here we introduce a new framework that characterises the similarity between networks using a multidimensional approach, which brings to light different ways in which networks can be (dis)similar and --- most importantly ---  complementary to each other. Our approach draws inspiration from the literature on Partial Information Decomposition (PID), which develops the principle that information is not a monolithic entity but can be disentangled into qualitatively different types, including redundancy (information that can be independently retrieved in more than one source), unique (information that can only be retrieved from a specific source), and synergy (information that cannot be retrieved from a single source, but only by taking into account multiple sources at once)~\citep{williams2010nonnegative, mediano2021towards, luppi2022natureneuro}. Inspired by the PID framework, in this paper we introduce the \textit{Partial Network Decomposition} (PND): a formalism that disentangles the relationship between networks in terms of different `similarity modes,' quantifying the extent to which two or more networks are redundant, unique, or synergistic. To illustrate the broad applicability of this framework, we showcase how PND provides new insights into real-world artificial and biological networks by analysing two scenarios: London's public transportation systems, and the topologies of short- vs long-range connections in the brains of humans and over 100 other mammalian species.

\section*{Results}

\subsection*{Conceptualising synergy and redundancy of networks}

Our central question is whether two networks built over the same nodes can be considered synergistic or redundant, or whether they provide unique contributions. Intuitively, two networks can be said to be redundant if the combination of the two is in some sense equivalent to each one of them separately. Conversely, one could say that they are synergistic if, when considered together, the two networks complement each other in some sense. 
Let us illustrate this principle in the case of transport networks: when considering the topologies of e.g. bus and train networks, synergy occurs if the best way --- in terms of cost or time --- to move between stations involves using both modalities. In contrast, the two networks are redundant if each provides alternative routes of equal efficiency, and their combination does not provide additional gains.

To capitalise on this intuition, 
let us consider two undirected binary connected networks (with no isolated nodes\footnote{This is only for simplicity of exposition --- the framework as outlined in Materials and Methods applies to disconnected networks as well.}) defined on the same set of nodes, denoted by $A$ and $B$ (Fig.~\ref{fig_ToyNets}). 
One can identify synergy between networks $A$ and $B$ whenever the most efficient (shortest) path between two nodes $x$ and $y$ involves traversing a combination of edges from $A$ and $B$. Practically, this means that the most efficient path between $x$ and $y$ is found on the joint network, constructed by placing an edge between each pair of nodes that are directly connected in either network $A$ or network $B$, such that the joint is also a binary undirected network. One can then operationalise efficient paths between nodes $x$ and $y$ over networks $A$ and $B$ as being redundant if they are of equal length --- such that a traveller would be indifferent between traversing via network $A$ or $B$. 
\footnote{Note that this differs from having multiple paths between $x$ and $y$ within network $A$ or network $B$~\citep{di2012redundancy}: what we care about is that $x$ and $y$ be reachable in the same number of steps within network $A$ and within network $B$.} Finally, if the most efficient path between nodes $x$ and $y$ is shorter in network $A$ (respectively, $B$) than in network $B$ (resp., $A$), then it is natural to label it as a ``unique'' contribution of network $A$ (resp., $B$). 
Thus, if the most efficient path between nodes $x$ and $y$ is of length $l_A$ when only using edges from network $A$, of length $l_B$ when only using edges belonging to network $B$, and of length $l_{A\cup B}$ when using edges from $A$ and $B$, then the classification of the link between $x$ and $y$ can be done according to the following procedure:
\begin{itemize}
    \item Synergistic if $\min\{l_A,l_B\} > l_{A\cup B}$.
    \item Unique if $\min\{l_A,l_B\} = l_{A\cup B}$ and also $\max\{l_A,l_B\} > l_{A\cup B}$.
    \item Redundant if $\max\{l_A,l_B\} = l_{A\cup B}$.
\end{itemize} 
It is direct to verify that this procedure identifies the most efficient (shortest) path between a pair of nodes $x$ and $y$ as being synergistic, redundant, or unique to either network $A$ or $B$, with no other outcome being possible. 

Following this rationale, we can obtain a global quantification of the prevalence of synergistic, redundant, and unique paths across the two networks in terms of the proportion of all shortest paths that they respectively account for. This procedure answers the following question:   how many of the possible maximally efficient paths between nodes would an agent know if they knew only network $A$, or only network $B$, or both?

\begin{figure*} 
\centering
\includegraphics[width=0.7\textwidth]{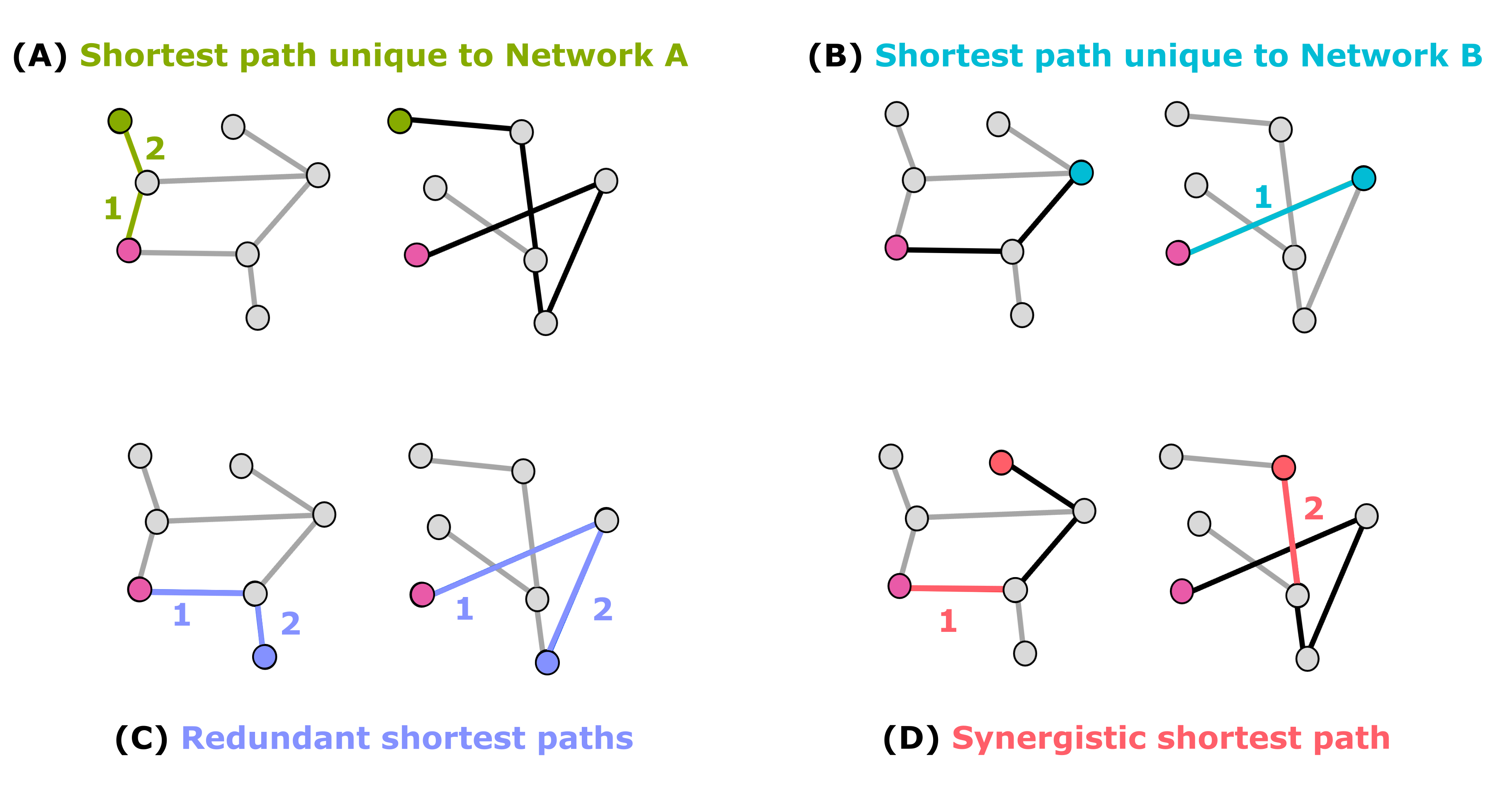}
\caption{{\bf Illustration of the efficiency partial network decomposition between two networks.} 
Green and light blue edges indicate shortest paths that are unique to each network (corresponding path in the other network is shown in black). Dark blue indicates redundant paths, which are equally long in either network; red indicates a synergistic path: shorter when the two networks are combined, than in either individual network (shown in black).}
\label{fig_ToyNets}
\end{figure*}

In addition to a global-level assessment, our approach also allows us to obtain further insight by focusing on the relevance of different scales (i.e., paths of different length). 

In effect, it is possible that the interactions between the two networks may be different for paths of different length --- e.g. short paths may be more redundant while longer ones are more synergistic. 
Our method provides a straightforward way to obtain such insight: since it provides a decomposition of the efficiency between each pair of nodes, one can simply group the network's shortest paths in terms of their length $l$, and check the proportion of them that are synergistic, redundant, or unique.\footnote{Note that length-1 paths (which correspond to direct edges) can only be, by definition, either redundant or unique.}

Finally, we emphasise that this running example of shortest paths in two networks is merely a special case of the more general Partial Network Decomposition introduced here. The full framework is applicable to any network measure of interest and any number of networks, and is described in detail in Materials and Methods.

\subsection*{Partial network decomposition in random network models}

To build some initial intuitions on how different networks may contribute to the most efficient paths of a combined network, we constructed pairs of binary undirected Erd\H{o}s-R\'enyi networks of different densities (ranging from $1\%$ to $100\%$), and evaluated the synergistic, redundant, and unique contributions between them. Results show that shortest paths become less synergistic and more redundant as the density of links grows (Fig.~\subref{fig_ErdosRenyi}{A}). In effect, the majority of maximally efficient paths on the joint network are synergistic when the networks are both sparse (i.e. both with densities of $5\%$ or less). Synergy is also present up to approximately $15\%$ density, thereafter dropping off rapidly. Unique contributions from one network tend to dominate when the other network is below approximately $15\%$ density, thereafter also levelling off rapidly. Thus, when the two networks' densities are imbalanced, unique contributions from the denser one tend to predominate. Conversely, when the two networks have similar density (and both above $15\%$ density), then the majority of the maximally efficient paths are redundant, with redundancy increasing gradually with both networks' density (Fig.~\subref{fig_ErdosRenyi}{B}).

To further develop intuitions, we also used partial network decomposition to investigate the small-world character of networks~\citep{watts1998nature}. Our setup considered two copies of the same lattice network, and progressively randomly rewired one of them by $1\%$ increments while preserving the degree sequence. For each level of randomisation, we calculated the efficiency partial network decomposition between the original lattice network and its randomised counterpart (Fig.~\subref{fig_ErdosRenyi}{C}).
Our results reveal two clearly distinct regimes. At low levels of randomisation, redundancy dropped very rapidly as randomisation increased, replaced by a prominent unique contribution from the rewired network, and also an increasing prevalence of synergy between the two. Then, after a threshold of randomisation has been surpassed (approximately $9\%$ in Fig.~\subref{fig_ErdosRenyi}{C}), the unique contribution from the randomised network reached its peak and began to decline, while redundancy slows its decline and began to plateau (Fig.~\subref{fig_ErdosRenyi}{C}). In contrast, both synergy and the unique information of the non-rewired lattice grow consistently with the degree of rewiring. 
Remarkably, we observed that the small-world propensity index~\citep{muldoon2016scirep} (see Methods) of the joint network peaks at the same point of the unique information of the rewired network.
\footnote{Note that since the joint network is obtained by combining the original lattice and its randomised counterpart, its density increases as the randomisation means that there are more and more non-overlapping edges. However, the small-world propensity is unaffected by density \citep{muldoon2016scirep}.}

Taken together, these results show that partial network decomposition can provide rich insights about the relationship between two networks. Building on these insights, in the following sections we show results of this machinery applied to real-world networks, both artificial and biological.

\begin{figure*} 
\centering
\includegraphics[width=0.8\textwidth]{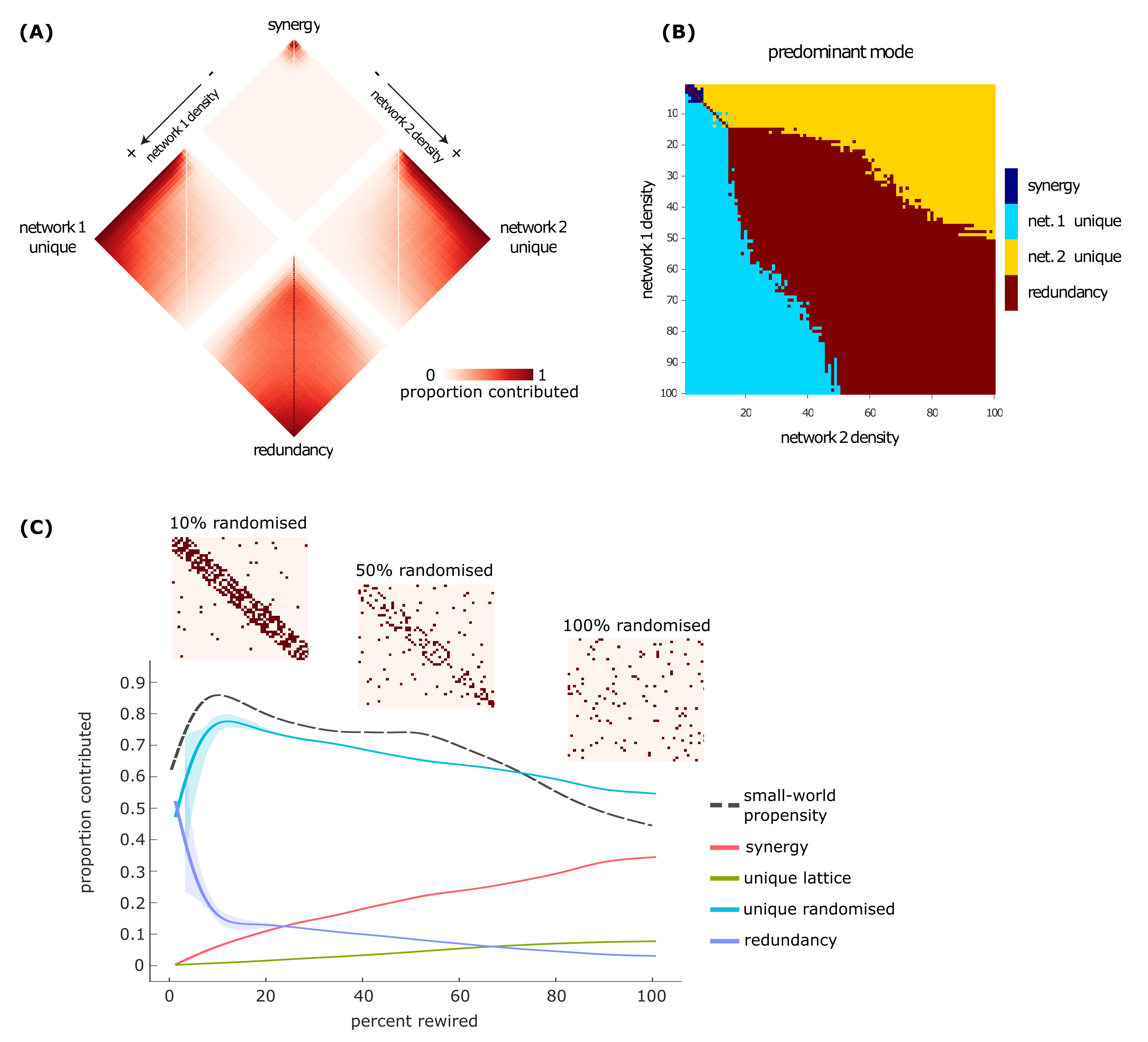}
\caption{{\bf Contributions to shortest paths as a function of network density and rewiring.} 
(A) Erd\H{o}s-R\'enyi networks: each contribution is shown as a function of both networks' densities. (B) Erd\H{o}s-R\'enyi networks: component accounting for the majority of shortest paths, as a function of both networks' densities. (C) Lattice and randomised networks: contributions are shown as a function of the percentage of rewired edges.}
\label{fig_ErdosRenyi}
\end{figure*}


\subsection*{London transport network}

We analysed real data from a context that may be familiar to everyday life --- transportation networks. We investigated the topologies of two means of transportation in London: underground (i.e. subway) and overground (including different types of local trains) (Fig.~\subref{fig_London_tube}{A}). 
We used PND to gain insight on how these two types of transportation serve the needs of Londoners to move within the city.

Our analyses reveal that there is extremely low redundancy between the topologies of underground and overground. Other similarity modes, however, exhibit a strong dependence on the length of the path. For example, short paths show an approximately $70\%$ unique contribution from underground and nearly $30\%$ contribution from overground. Synergy rises rapidly with path length, accounting for over $60\%$ of all paths at its peak, and over $50\%$ of paths of most lengths. This changes rapidly at the very longest paths, where again underground unique paths predominate, eventually reaching $100\%$ contribution (Fig.~\subref{fig_London_tube}{B}).
Overall, these findings speak about the efficiency of the design of these networks, which serve the city with almost no redundancy and substantial synergy.

To evaluate the significance of these findings, we compared the obtained results against the decomposition arising from a null distribution involving randomly rewiring both networks, while preserving the degree sequence to account for the potential confounding effects of this low-level network property (Fig.~\subref{fig_London_tube}{C}). 
We found that the London transport network relies on unique contributions from the underground network significantly more than would be expected by chance ($p = 0.007$). Intriguingly, although redundancy is by far the least prevalent term in the decomposition, its value is nevertheless significantly greater ($p < 0.001$) than what would be expected based on two random networks of equal density and degree sequence (Fig.~\subref{fig_London_tube}{C}). The other two contributions (overground unique and synergy) did not significantly differ from their degree-preserving null counterparts. However, when null distributions were obtained using purely random networks (with same density as the original ones) instead of degree-preserving random networks,
the unique contribution of the overground network was also significantly greater than expected from chance ($p < 0.001$), whereas the contribution of synergy was significantly lower than expected from chance ($p = 0.002$). Thus, we see that the degree sequence plays a role in determining the decomposition. 

\begin{figure*} 
\centering
\includegraphics[width=0.8\textwidth]{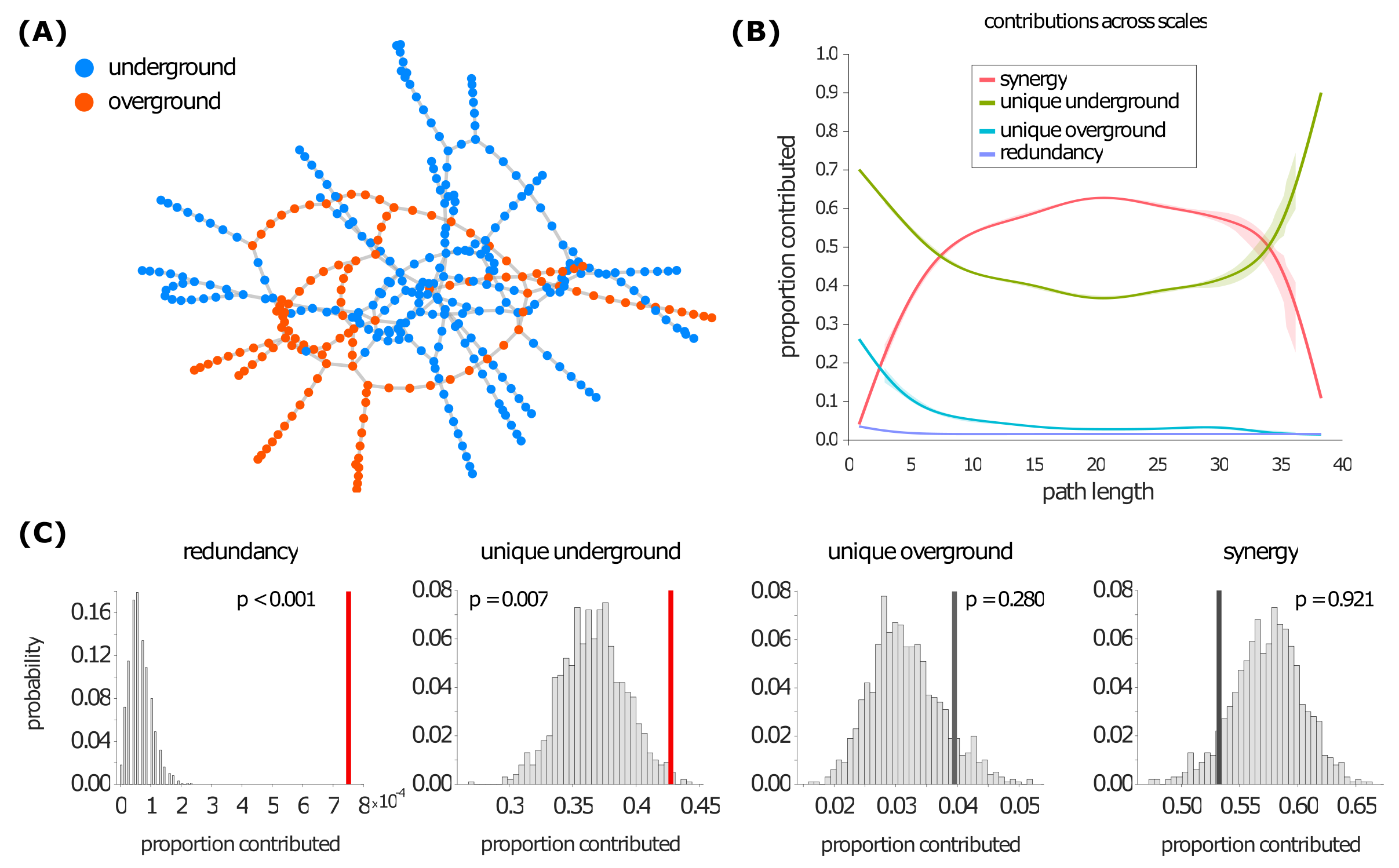}
\caption{{\bf London transport networks.} 
(A) The underground network (blue) and the overground network (orange). (B) The relative proportion of shortest paths accounted for by redundant, unique, and synergistic contributions, as a function of path length. (C) Empirical results (vertical lines) against a null distribution obtained from degree-preserving randomisation of the two networks. Red line indicates statistical significance. }
\label{fig_London_tube}
\end{figure*}

\subsection*{Connectivity networks in the human brain}

The next step in our analysis was to investigate the relationship between the networks of white matter fibers that link spatially proximal regions (i.e., short-range fibers) and spatially distant regions (i.e., long-range fibers) within the human brain (note that we use ``short-range'' and ``long-range'' to refer to the physical (Euclidean) distance between the two regions they connect --- not to be confused with the shortest path between regions, which is in terms of the number of hops on the network). White matter fibers provide the anatomical scaffold over which communication unfolds in the brain; understanding how their networked organisation supports brain function is a major research topic in neuroscience~\citep{suarez2020tics}. 

We used the PND framework to investigate the relationship between the efficiency of networks of white matter tracts connecting spatially proximal and spatially distant regions of the human brain, reconstructed from in-vivo diffusion MRI tractography in $100$ healthy human adults (see Methods). For each subject, we defined one network as comprising the $50\%$ of fibers connecting the most spatially distant regions in that subject's brain, in terms of greatest Euclidean distance; and a second network as comprising the $50\%$ of white matter fibers connecting the most spatially proximal regions (smallest Euclidean distance between them). Thus, the two networks have equal density. As a first (subject-level) analysis, we decomposed the similarity of connections between spatially proximal and spatially distant regions, observed in each subject. 

Results revealed that, on average across individuals, nearly $50\%$ of the maximally efficient paths in the overall structural connectome (combining all white matter tracts) are accounted for by long-range connections between spatially distant regions, while synergistic paths are the second-largest contributors (Fig.~\subref{fig_SC_longshort}{A}). This suggests that long-range fibers play a key role in enabling communication between regions that are distant not only physically, but also topologically (the most efficient path between them involves many hops), despite being more metabolically expensive. 
In addition to this high-level description, however, our approach can also provide more detailed information. We began by considering paths of different length: our framework revealed that synergistic contributions become prevalent for paths comprising multiple hops (Fig.~\subref{fig_SC_longshort}{B}). This may be expected, since the longer the path, the more occasions there may be for making it more efficient with an appropriately-placed connection. 

\begin{figure*} 
\centering
\includegraphics[width=0.98\textwidth]{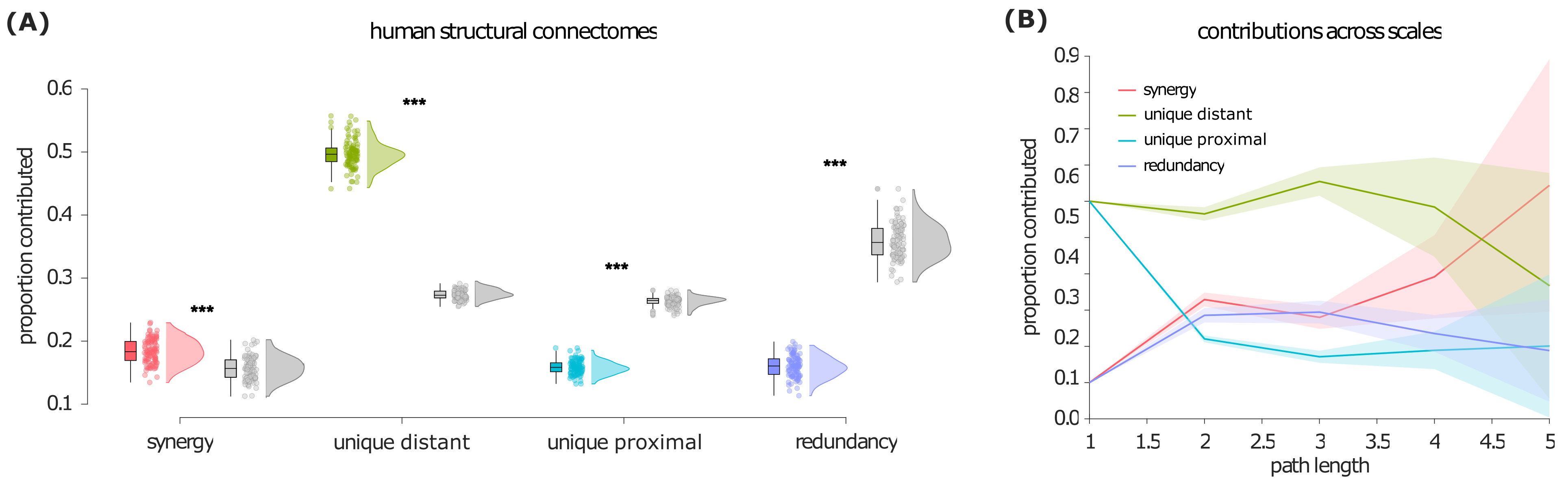}
\caption{{\bf Prevalence of synergistic, unique, and redundant efficiency contributions, for networks of white matter fibers between spatially proximal and spatially distant regions of the human brain.} 
(A) Proportion of the most efficient paths accounted for by each PND term. Box-plots indicate the median and inter-quartile range of the distribution. Each data-point is one subject (N=100). Grey distributions indicate the corresponding values for degree-preserving randomised networks. ***: $p < 0.001$ against null distribution of values obtained from rewired null networks.  (B) Prevalence of synergistic, unique, and redundant contributions as a function of path length, for human structural brain networks.}
\label{fig_SC_longshort}
\end{figure*}

To gain more insight on these results, we repeated our decomposition on consensus networks obtained using a procedure to aggregate individual connectomes~\citep{betzel2019netneurosci}, which provides two networks that are representative of the topology of long-range and short-range white matter fibers in the human brain (see Methods). 
By decomposing the similarity of these representative pair of networks, we see again that approximately $50\%$ of the maximally efficient paths in the network are accounted for by white matter tracts between spatially distant regions, confirming the results on individual subjects discussed above (Fig.~\subref{fig_SC_longshort_nets}{A}). Importantly, these analyses allow us to study these networks at the level of individual edges. Results show that regions best reached via paths along long-range fibers are predominantly located in different hemispheres, and at opposite ends of the anterior-posterior axis (Fig.~\subref{fig_SC_longshort_nets}{B}). Cross-hemisphere connections are also prominent for synergistic paths. In contrast,  redundant and uniquely short-distance paths are primarily located within the same hemisphere (Fig.~\subref{fig_SC_longshort_nets}{B}). This is to be expected, since connecting physically distant nodes by traversing short-distance fibers inevitably involves many hops, making this a suboptimal strategy in terms of minimising path length.

To demonstrate the robustness of our results, we show that the same pattern, with white matter tracts between spatially distant regions accounting for the most of the maximally efficient paths, can be replicated in an independent dataset of human diffusion MRI, which used Diffusion Spectrum Imaging (Supp. Fig.~\ref{fig_Lausanne_S1000_replication}) and defined brain regions using an alternative anatomical parcellation, including subcortical structures. Likewise, we replicate the human results with a higher-resolution version of the Schaefer cortical parcellation ($1000$ nodes); in this case, we see an even stronger contribution of synergy --- though still second to the unique contribution of long-range tracts (Supp. Fig.~\ref{fig_Lausanne_S1000_replication}). 

\begin{figure*} 
\centering
\includegraphics[width=0.8\textwidth]{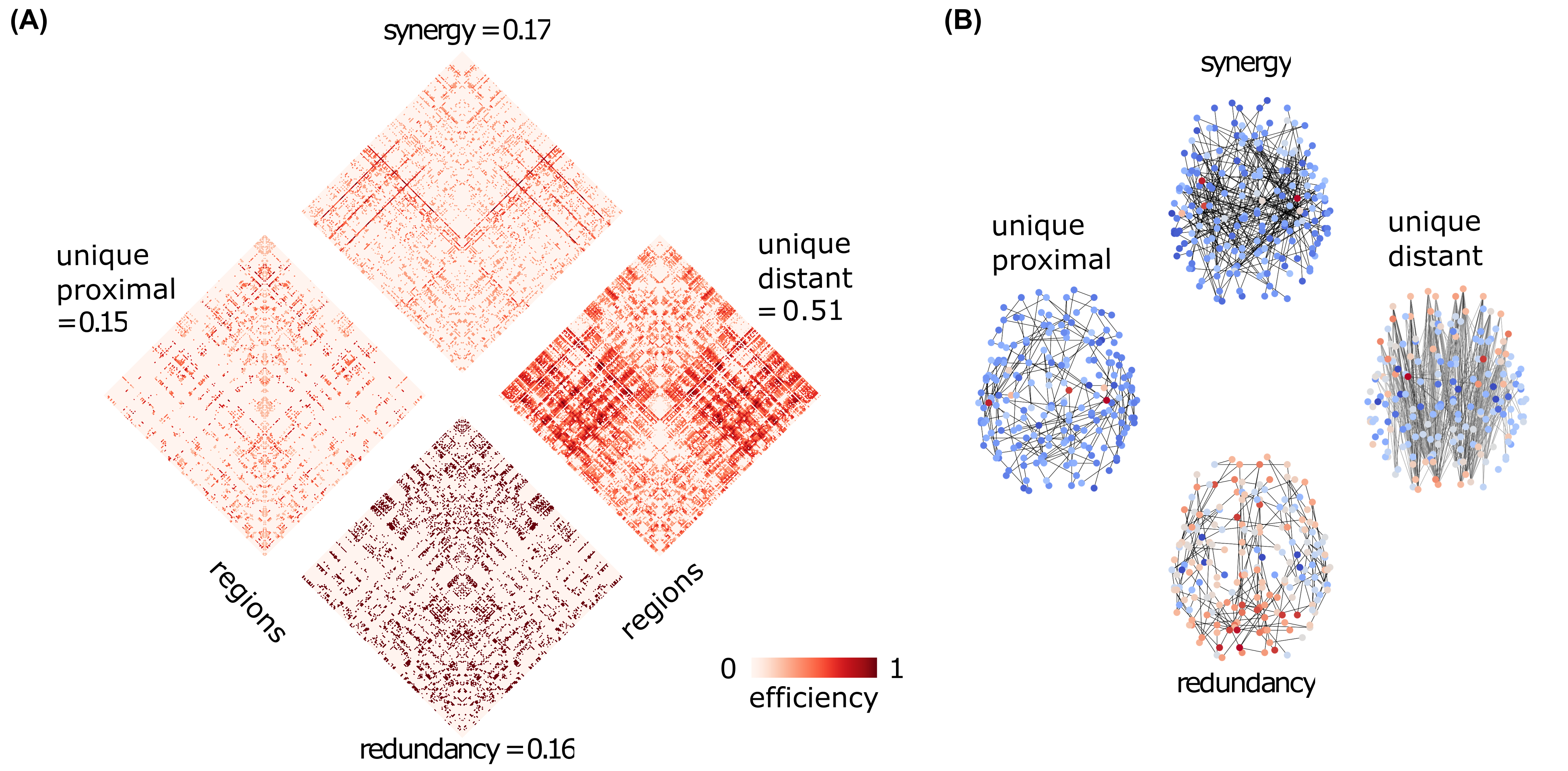}
\caption{{\bf Edge-wise decomposition into synergistic, unique, and redundant efficiency contributions, for group-consensus networks of white matter fibers between spatially proximal and spatially distant regions of the human brain.} (A) Each edge is assigned to the network of the corresponding mode, so that each edge is only non-zero in one of the four networks, and its value reflecting the gain in path against the next best alternative. For each matrix, upper and lower quadrants correspond to inter-hemispheric connections, and right and left quadrants are inter-hemispheric. (B) Same as (A), but with edges plotted in the brain to highlight distinct patterns of connectivity in the human brain. Warmer colour indicates higher node degree. Networks are thresholded for visualisation purposes.}
\label{fig_SC_longshort_nets}
\end{figure*}

We then investigated the role of network topologies in shaping their respective contributions. For this, we consider null models obtained by degree-preserving randomisation of the original structural connectomes. Examining rewired networks enables us to assess whether the results for the human structural connectomes could be observed just by chance for any random network with the same density and degree distribution. We see that if the networks are randomly rewired (while still preserving the degree sequence), synergy becomes the lowest contributor, and redundancy is the highest (Fig.~\subref{fig_SC_longshort}{A}, grey plots). This is clearly distinct from the predominance of long-range connections observed in the empirical structural connectomes. Statistical comparisons confirm that the human connectome relies on white matter tracts between spatially distant regions, and synergy between long- and short-range tracts, significantly more than a random network would. On the other hand, the human connectome makes significantly less use of connections between spatially proximal regions, and especially redundant communication pathways (Table~\ref{table_SC_vs_Nulls}).

 Next, we considered empirical brain networks obtained from a different neuroimaging modality: correlation of functional MRI BOLD signals (i.e., `functional' connectomes). Unlike diffusion MRI, functional connectivity reveals connections between regions in terms of the similarity of their activity over time, thereby providing a different perspective on the network organisation of the human brain. This enables us to ask whether any empirical brain networks will display the same pattern of results reported above, or whether those results are specific to the \textit{structural} connectome.

Again, we consider one network of (functional) connections between spatially proximal regions, and one network of connections between spatially distant regions, for each individual. What we find is that for functional connections, synergy is by far the main contributor (Fig.~\ref{fig_FC_longshort}) --- although exhibiting high variability across subjects. Notably, these results in functional networks greatly differ from the findings in structural networks
(both original and randomly rewired). 

Thus, structural connectomes are dominated by unique contributions from fibers between spatially distant regions; whereas functional connectomes are dominated by synergistic combinations of connections between proximal and spatially distant regions.

\begin{figure} 
\centering
\includegraphics[width=1\columnwidth]{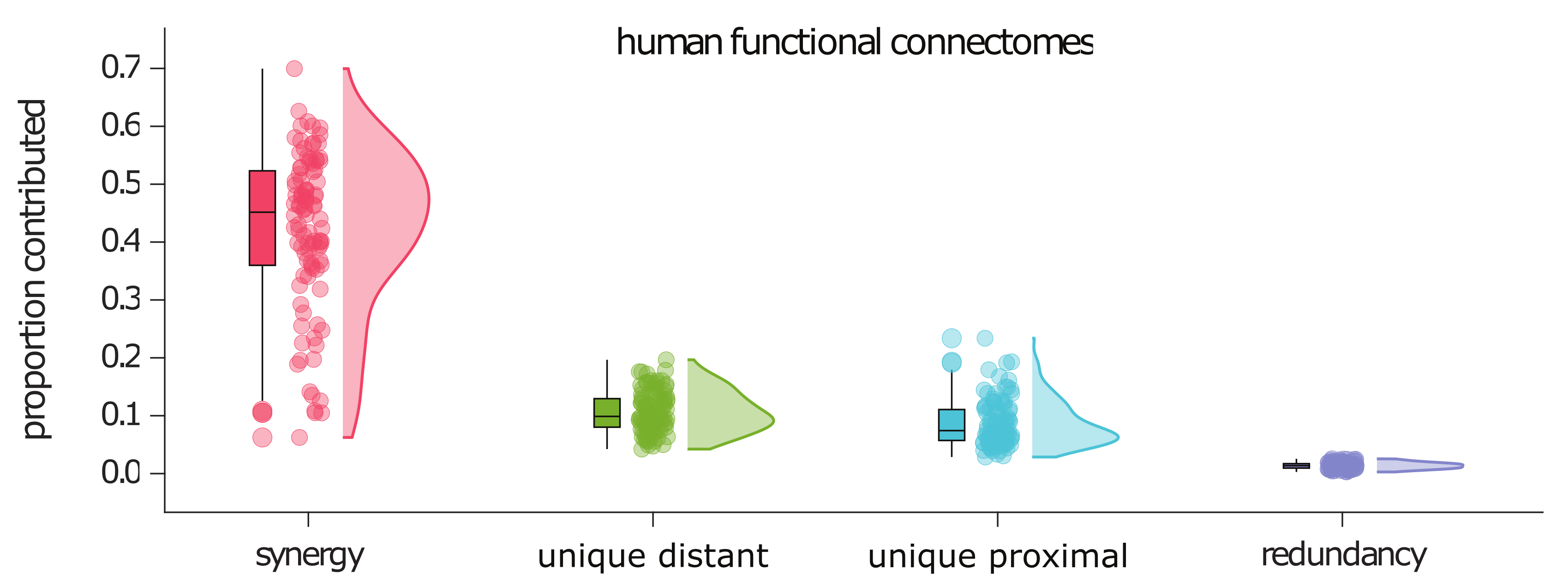}
\caption{{\bf Prevalence of synergistic, unique, and redundant efficiency contributions, for networks of functional connectivity between spatially proximal and spatially distant regions of the human brain.} Y-axis: proportion of shortest paths accounted for by each PND term. Box-plots indicate the median and inter-quartile range of the distribution. Each data-point is one subject (N=100).}
\label{fig_FC_longshort}
\end{figure}

\subsection*{Structural brain networks across mammalian species}

To evaluate if the obtained results are distinctive of the human structural connectome or if this is also observed in other species, we performed analogous analyses over a wide spectrum of structural connectomes from diffusion MRI, 
covering $N=220$ individual animals from $125$ mammalian species~\citep{assaf2020conservation, suarez2022connectomics} (see Methods). 

We performed the same analysis as for the human structural connectomes, using Euclidean distance to divide white matter tracts into those linking spatially proximal regions, and those linking spatially distant regions (each accounting for $50\%$ of the total edges, thereby ensuring two equally dense networks). We then identified the synergistic, redundant, and unique contributions to the global efficiency of the joint network. Similarly to the human structural connectome (and unlike the human functional connectome), our results show that a substantial proportion of the most efficient paths in the network are accounted for by white matter tracts between spatially distant regions (Fig.~\ref{fig_MAMI_violins}). Results also show that mammalian structural connectomes are significantly more synergistic and more reliant on white matter tracts between spatially distant regions, than corresponding randomised null models (see Fig.~\ref{fig_MAMI_violins} and Supp. Table~\ref{table_MAMI_vs_Nulls}). Unlike the human case, however, we find that non-human mammals also involve a substantial contribution of redundancy --- though still significantly less than randomly rewired nulls.

\begin{figure*} 
\centering
\includegraphics[width=0.98\textwidth]{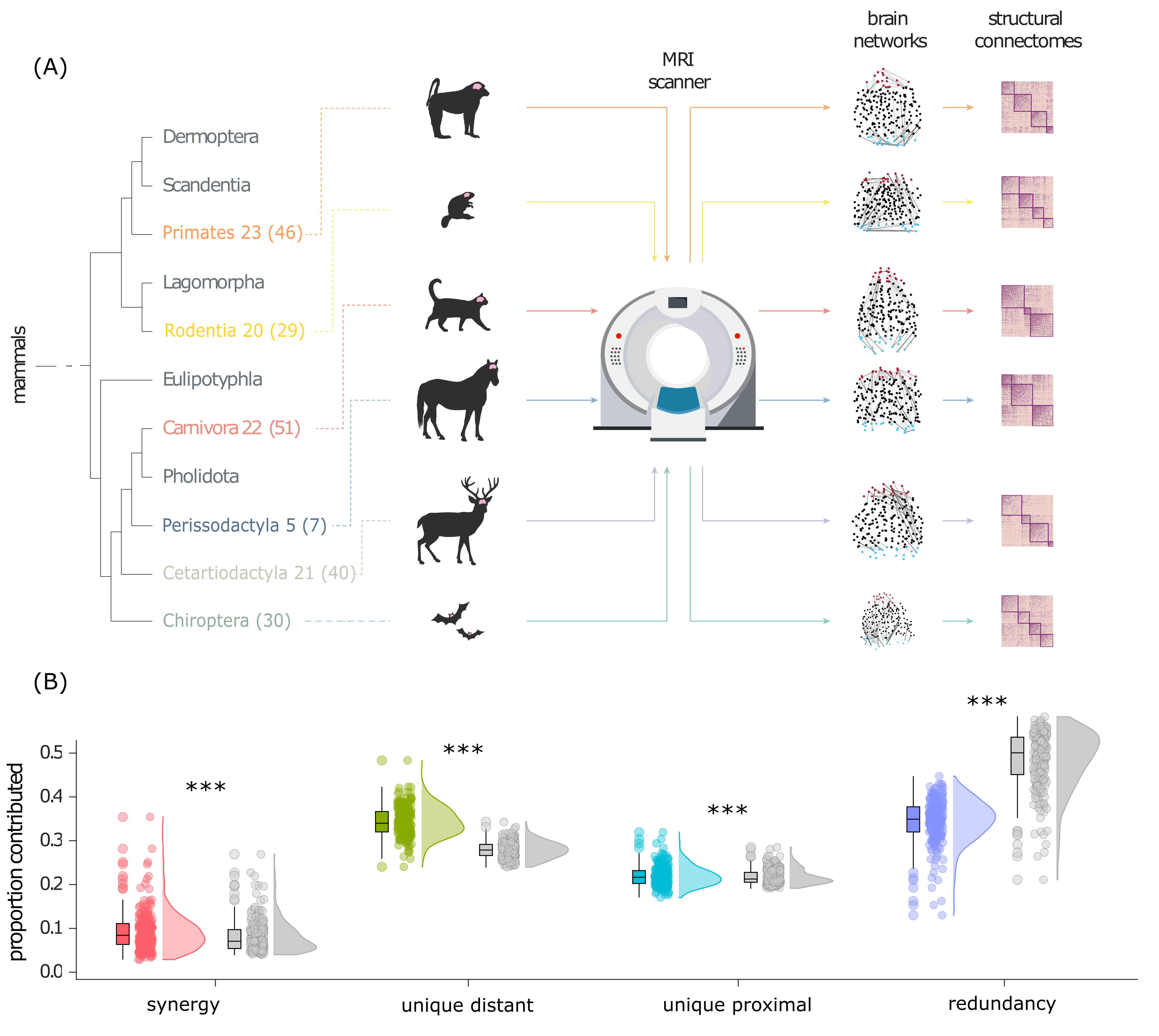}
\caption{{\bf Prevalence of synergistic, unique, and redundant efficiency contributions, for networks of white matter fibers between spatially proximal and spatially distant regions of $\bm{220}$ mammalian brains.} 
Grey distributions indicate the corresponding values for degree-preserving randomised networks. ***: $p < 0.001$ against null distribution of values obtained from rewired null networks. Y-axis: proportion of shortest paths accounted for by each PND term. Box-plots indicate the median and inter-quartile range of the distribution. Each data-point is one animal (N=$220$), or its randomised null.}
\label{fig_MAMI_violins}
\end{figure*}

\section*{Discussion}

Drawing inspiration from the field of information decomposition, here we introduced a simple but powerful framework to quantify qualitatively different aspects of the similarities and differences between two networks in terms of a graph measure of interest. When applied to global efficiency, our approach determines the contribution of each network to the shortest-path communication on the joint network, assessing whether such contributions are redundant, unique to each of the two initial networks, or synergistic --- a novel metric of the complementarity between two networks such that combining them makes some of the shortest paths even shorter. 
By considering shortest paths, this approach takes into account the topology of the entire networks. It can also provide insights at different levels of resolution: from summary statistics (proportion of all shortest paths contributed by each mode); to information about their relative prevalence across paths of different length; to edge-level detail including the number of steps saved with respect to the next-best alternative.

Our analyses illustrate how this approach can be applied to deepen our understanding of brain networks as well as transportation networks. In the case of London's transport networks, redundancy was found to be low in absolute terms, but nevertheless significantly greater than that of random networks.
In brain networks, partial network decomposition highlighted stark differences between the roles of long-range and proximal connections. We observed that reliance on long-range white matter fibers for the majority of efficient paths in the network is a conserved feature of structural connectomes across mammalian species, including humans. This feature is not shared by just any brain network (it is not found in human functional connectomes), nor by random networks having the same degree distribution. On the contrary, we found that the brains of both humans and other mammals also exhibit significantly more synergy, and significantly less redundancy, than corresponding randomly rewired networks. 

It is important to note that, unlike traffic networks, the brain does not rely solely (or perhaps even primarily) on communication via shortest paths, instead likely adopting a variety of mechanisms \citep{seguin2018navigation, seguin2019inferring, avena-koenigsberger2017natrevneurosci}. Therefore, our analysis of connectomes is not intended to quantify how signals actually propagate between brain regions --- rather, our intention was to investigate what roles the topology of proximal and long-range connections in the brain would play in determining shortest paths. 

In this sense, our analysis should be seen as analogous to the numerous investigations of brain small-worldness and global efficiency~\citep{assaf2020conservation, bassett2017natneuro, muldoon2016scirep}, which assess the network's suitability for shortest-path communication without claiming that such communication in fact happens. This approach is, therefore, distinct from studies that evaluate the role of different communication strategies for explaining empirical patterns of inter-regional communication \citep{vazquezrodriguez2019pnas, suarez2020tics}.

By studying random network models, the proposed framework showed that synergy generally predominates when both networks are very sparse ($~ 5\%$ density or less)\footnote{Note that all our comparisons were performed against null models that preserved the original networks' density and degree, ensuring that such low-level properties do not explain our results.} Conversely, synergy rapidly diminishes as the density of links in either network grows. This observation is noteworthy because a previous survey suggested that most biological and human-made networks are indeed sparse~\citep{melancon2006sparse}, which according to our findings could be conducive to synergy. 

Although it is apparent why synergy between networks can be advantageous, it is important to bear in mind that redundancy can also be valuable. The presence of redundant paths between nodes on two different networks means that transport will be unaffected, even if one of the two networks should fail. In contrast, failure of either network would easily jeopardise a synergistic path. In other words, redundancy of networks facilitates robustness, echoing previous findings in time series analysis~\citep{luppi2021networkneuro}. 

\subsection*{Future directions}

Several future extensions of this framework are possible. First, we only considered unweighted, undirected networks. The PND method can be directly applied to weighted networks, provided that the weights can be combined across the two source networks --- as would be the case e.g. for prices or time, in the context of transport networks. Likewise, directed networks can also be straightforwardly accommodated. 
It is also worth acknowledging that transitioning between networks at a given node is not always possible, and not always free of cost, and this may need to be taken into account for future extensions. As an example: if the cost is time (for transport, for instance), then the time spent while waiting between different transport networks may need to be taken into account (e.g., airport layovers); whereas in terms of ticket price, there is often no cost for switching between e.g. different metro lines, but there can be a price for switching between metro and train. This would need to be taken into account when evaluating the advantage conferred by synergy between the two networks. 

Our approach requires a choice of how to operationalise the notion of ``path redundancy.'' While in this work we have adopted equal length of shortest paths as an intuitive criterion for redundancy, this is not the only possibility. An interesting alternative is in terms of identity of edges traversed: will there be a path of length $l$ between nodes  $x$  and $y$ even if edge $q$ is taken out? This may contribute to the characterisation of redundancy as robustness. However, note that the present framework does not take into account whether multiple equally short paths between nodes  $x$  and $y$ exist within network $A$. The existence of such redundant shortest paths within a network has also been proposed as a metric of robustness, though distinct from the one developed here \citep{stanford2022pnas}. Accounting for redundant shortest paths both within and between networks may provide a fruitful avenue for future extensions of the present work.  

Also, the present framework relies on a measure of costs --- here we used the global efficiency, which is based on shortest paths. However, different communication protocols can exist on networks: whereas shortest-path navigation requires knowledge of the global topology, other approaches can be agnostic, such as network diffusion based on random walks, or navigability~\citep{seguin2018navigation, seguin2019inferring, avena-koenigsberger2017natrevneurosci}. Incorporating different communication protocols over networks will be a valuable extension of our framework. One reason that shortest paths are especially appealing for our approach is that, when combining two networks (predicated on the same nodes, i.e. layers of a multiplex), shortest paths can never become longer: only grow shorter (if there is synergy between networks) or remain the same as the shortest of the two. Thus, the path efficiency can only increase or stay the same. However, this property is not guaranteed in the context of diffusion: unless the two networks are identical, the joint network will be denser than either of them, meaning that on average a random walker will have more options to choose between, when leaving node $x$: although a shortcut to node $y$ may now exist, the average number of steps to reach node $y$ may still increase (because the random walker has more chances of choosing an edge that does not belong to the shortest path – and this repeats at every new node). Therefore, in its present form our framework may often lead to no synergy being identified, when considering communication via random walks. 

Relatedly, here we did not consider the question of what dynamics take place over the network \citep{harush2017dynamic, kirst2016dynamic}. For instance, under dynamics that allow for the possibility of congestion (e.g., traffic), adding a shortcut may not always make travel on the network faster. On the contrary, it may even achieve the opposite effect by increasing congestion --- a phenomenon known as ``Braess`s paradox" \citep{braess1968paradoxon}. Thus, two networks that are synergistic in terms of path length may nevertheless end up having ''detrimental synergy'' in terms of travel time, depending on the dynamics of navigation on the network. It will be a fruitful topic of future investigation to investigate such scenarios using our framework.

As a limitation, we note that our use of binary networks required us to threshold the functional connectivity (see Methods), which imposes a somewhat arbitrary criterion (though not devoid of ground \citep{luppi2021networkneuro}). However, this allowed us to keep the same edge density in the FC and SC networks, thereby removing the influence of density from our results. However,  for the London transport example the two networks had different density of edges (underground vs overground); this influences the prevalence of unique paths, since in a binary network an edge corresponds to the shortest (i.e., maximally efficient) possible path between the two nodes at its extremes, so a denser network will have more short paths (whether unique or redundant). Indeed, we found it to be so - but to an extent that was statistically unexpected based on density alone. Future work may extend the present framework to quantify the extent to which a network uses its available connections synergistically, uniquely, or redundantly, as a proportion of the theoretical maximum for a given partner network.

In conclusion, we have developed a simple yet versatile approach to characterise the relationship between two networks, taking into account topological properties and capable of providing both global and local insights. This work may find application when engineering network systems, to help decide how to trade-off different desiderata (such as efficiency vs robustness introduced by a new underground line); or evaluate their relative presence in existing systems, as we have done here for the human brain.

\section*{Materials and Methods}

\subsection*{Mathematical formalisation}

In this section we provide the mathematical foundations of our multiplex network analysis method. It is based on ideas from information theory, which we introduce below, and relies on a mapping between graphs and probability spaces. Beyond the specific metrics showcased in the results of this paper, this formalism paves the way for very general analyses of multiplex networks by combining principles of probability, graphs, and information theory.

\subsubsection*{Mapping probability theory and graphs}
\label{sec:mapping}

A graph can be defined as 
a pair $(\mathcal{V}, \mathcal{E})$, where $\mathcal{V}$ is a set of vertices (or nodes) and $\mathcal{E} \subseteq \mathcal{V} \times \mathcal{V}$ is a set of edges indexed by pairs of vertices. 
Our analyses focus on 
a `utility function' $f$, which corresponds to a network property of interest. An example of such function is the network's global efficiency, which is the average of the efficiency between each pair of nodes. For a given node pair $\omega = (v_1, v_2) \in \mathcal{V} \times \mathcal{V}$ in a network with edges $\mathcal{E}$, efficiency is defined as the inverse of the length of the shortest path between $v_1$ and $v_2$, and we denote it by $f(\omega; \mathcal{E})$. 
Additionally, we consider a probability measure $p(\omega)$ uniform over pairs of distinct nodes $\omega\in \mathcal{V}\times\mathcal{V}$, which establishes a random variable $\Omega$. With this, the average value of the utility function $f$ over the network is given by
\begin{align}
    F(\mathcal{E}) = \E{f(\Omega; \mathcal{E})} ~,
\end{align}
where $\E{f(\Omega; \mathcal{E})} = \sum_\omega f(\omega; \mathcal{E}) p(\omega)$ is the expected value operator. 
In the example above, if $f$ is the pairwise efficiency between two nodes, then $F$ corresponds to the global efficiency of the whole network. 

Now, suppose that we have two different sets of edges $\mathcal{E}_1$ and $\mathcal{E}_2$ for the same set of vertices $\mathcal{V}$, and we wish to compute how they affect the value of $f$. 

For example, a natural question is how much of the value of $f$ can be attributed to the edges in $\mathcal{E}_2$, over and above those in $\mathcal{E}_1$? This can be naturally estimated as
\begin{align}
    \Delta F(\mathcal{E}_1 | \mathcal{E}_2) = F(\mathcal{E}_1 \cup \mathcal{E}_2) - F(\mathcal{E}_2) ~ .
\end{align}
For example, if $F$ is the global efficiency, then $\Delta F(\mathcal{E}_1 | \mathcal{E}_2) \geq 0$ captures how much the efficiency increases by adding $\mathcal{E}_1$ to the links in $\mathcal{E}_2$. 

\subsubsection*{Simple example: Decomposing efficiency in two-layer networks}

Using the established link between graphs and probability theory, one can take inspiration from frameworks to decompose information-theoretic quantities. In particular, here we use ideas from the \textit{Partial Information Decomposition} framework~\cite{williams2010nonnegative}, and develop a new set of tools to decompose the impact of various layers of edges on a given observable. Before presenting the general formalism, here we illustrate the main ideas for a simple case of a two-layer network.

For this, let us consider a multi-layer network $\mathcal{M}$ with vertices $\mathcal{V}$ and two sets of edges (i.e. layers) $\mathcal{E}_1, \mathcal{E}_2$. The full (or joint) network $\mathcal{E}$ contains all edges of both layers, and mathematically is given by the union operator $\mathcal{E} = \mathcal{E}_1 \cup \mathcal{E}_2$. As an example, here we decompose the global efficiency of the joint network $F(\mathcal{E})$ into qualitatively different types of contributions from $\mathcal{E}_1$ and $\mathcal{E}_2$.

Intuitively, for any given pair of vertices $v_1,v_2$ there are four possibilities: the efficiency could be equal in both layers of the network, it could be greater in one layer than in the other, or it could be greater in the joint network than in either layer. We refer to these different cases as \emph{redundant}, \emph{unique}, and \emph{synergistic}, respectively.

Let us now present a plausible (and empirically powerful) definition for each of these contributions. For a given pair of vertices $v_1,v_2$ we can take their redundant efficiency (or just \emph{redundancy}) to be the minimum efficiency one could find by using either $\mathcal{E}_1$ or $\mathcal{E}_2$ -- that is, using the least efficient path fully contained within either of the layers. Mathematically, this can be written as
\begin{align}
    r(v_1, v_2; \mathcal{M}) = \min_i f(v_1, v_2; \mathcal{E}_i) ~ .
\end{align}
Based on this formula, the natural definition of the unique contribution of layer $\mathcal{E}_j$ is its gain in efficiency with respect to the redundancy for the same pair:
\begin{align}
    u_j(v_1, v_2; \mathcal{M}) = f(v_1, v_2; \mathcal{E}_j) - r(v_1, v_2; \mathcal{M}) ~ .
\end{align}
Note that, with this definition, for any given pair of nodes one layer will have zero unique contribution.

Finally, the synergistic contribution corresponds to the increase in efficiency seen in the joint network but not in either layer. Mathematically:
\begin{align}
    s(v_1, v_2; \mathcal{M}) =& f(v_1, v_2; \mathcal{E}) - \big(r(v_1, v_2; \mathcal{M}) \nonumber \\
    &+ u_1(v_1, v_2; \mathcal{M}) + u_2(v_1, v_2; \mathcal{M})\big) ~ .
\end{align}

To link back to our explanation above, we can take the expected value of these quantities with respect to $p(\omega)$ (which in the simplest case is an average across all pairs of nodes). This yields the average quantities $R(\mathcal{M}) = \E{r(\Omega; \mathcal{M})}$ and similarly for unique and synergistic contributions. With this, we can write
\begin{align}
    F(\mathcal{E}) = R(\mathcal{M}) + U_1(\mathcal{M}) + U_2(\mathcal{M}) + S(\mathcal{M}) ~ ,
\end{align}
proving that indeed we have decomposed average global efficiency into four constituent quantities. Note that the global efficiency depends only on the joint network but each atom depends on the full multi-layer network, since they depend on which edges are in $\mathcal{E}_1$ or $\mathcal{E}_2$.

To summarise the overall prevalence of synergistic, unique, and redundant paths in a network, we can define the dominant character of a given node pair $\omega = (v_1, v_2)$ as the highest-order non-zero atom in its partial
network decomposition. For example, in the two-layer case described here, we say that $\omega$
is a synergistic pair if $s(\omega; \mathcal{M}) > 0$; or a unique pair if
$u_j(\omega; \mathcal{M}) > 0$ and $s(\omega; \mathcal{M}) = 0$ (for any $j$);
or a redundant pair if $r(\omega; \mathcal{M}) > 0$ and $s(\omega; \mathcal{M})
= u_j(\omega; \mathcal{M}) = 0$ (for all $j$).

\subsubsection*{General framework of Partial Network Decomposition}

After presenting an elementary example,
let us introduce our full formalism, \textit{Partial Network Decomposition} (PND): an approach to multi-layer network analysis inspired by Partial Information Decomposition (PID)~\cite{williams2010nonnegative}.

Consider a set of nodes $\mathcal{V}$ and $N$ sets of edges $\{\mathcal{E}_i\}_{i=1}^N$, such that the tuples $(\mathcal{V}, \mathcal{E}_i)$ form networks with the same nodes but different edges. 
For a given set of indices $a = \{n_1,..,n_k\} \subseteq \{1,\dots,N\}$, we define the joint network $(\mathcal{V}, \mathcal{E}^a)$ with $\mathcal{E}^a = \bigcup_{i=1}^k \mathcal{E}_{n_i}$. Furthermore, let us denote collections of such networks by $\alpha=\{a_1,\dots,a_L\}$.
For example, possible collections of networks for $n = 2$ are $\{\varnothing\}, \{\{1\}\}, \{\{1\}, \{1, 2\}\}$, etc.

The key quantity in PND is the \emph{network redundancy function} $F_\cap(\mathcal{E}^{a_1}, \mathcal{E}^{a_2}, ..., \mathcal{E}^{a_L})$, 
which we will also denote with the shorthand notation $F_\cap^\alpha$.

This function should capture how much of the value of $F$ is due to the ``common contribution'' of all networks $\mathcal{E}^{a_1}, ..., \mathcal{E}^{a_L}$. We require this function to have the following properties:
\begin{description}

    \item[Symmetry] $F_{\cap}(\mathcal{E}^{a_1}, ..., \mathcal{E}^{a_L})$ is invariant to re-ordering of $\mathcal{E}^{a_1}, ..., \mathcal{E}^{a_L}$.

    \item[Self-intersection] $F_\cap(\mathcal{E}^a) = F(\mathcal{E}^a)$. This links PND with the original network property of interest, and is analogous to the set-theoretic statement that \mbox{$S \cap S = S$}.

    \item[Deterministic equality] $F_\cap(\mathcal{E}^{a_1}, ..., \mathcal{E}^{a_L}) = F_\cap(\mathcal{E}^{a_1}, ..., \mathcal{E}^{a_{L-1}})$ if $\mathcal{E}^{a_{L-1}} \subseteq \mathcal{E}^{a_L}$. This is analogous to the usual set-theoretic statement that $S \cap T = S$ if $S \subseteq T$.

\end{description}

In principle, one could apply $F_\cap^\alpha$ to any collection of networks $\alpha \in \mathcal{P}_1(\mathcal{P}_1(\{1,...,N\}))$, where $\mathcal{P}_1(S)$ is the power set of $S$ excluding the empty set. However, the deterministic equality axiom allows us to simplify the domain of $F_\cap$: for example, for the two-layer case we know that $F_\cap(\mathcal{E}^{\{1\}}, \mathcal{E}^{\{1,2\}}) = F_\cap(\mathcal{E}^{\{1\}})$, because $\mathcal{E}^{\{1\}}$ is contained in $\mathcal{E}^{\{1,2\}}$ (recall that $\mathcal{E}^{\{1,2\}} = \mathcal{E}_1 \cup \mathcal{E}_2$). In the general case, this means we can restrict the domain of $F_\cap^\alpha$ to the set of \textit{antichains} of $\{1,...,N\}$, denoted here as $\mathcal{A}$, which are naturally organised in a set-theoretic construct known as a \textit{lattice}~\cite{williams2010nonnegative}. For the two-layer case, the possible antichains are $\{\{1\},\{2\}\}$, $\{\{1\}\}$, $\{\{2\}\}$, and $\{\{1,2\}\}$.

Although the three properties are the only ones necessary to formulate our PND over the antichain lattice, there is one more property that will be of interest for the interpretation of the resulting decomposition:

\begin{description}
    
    \item[Monotonicity] $F_\cap(\mathcal{E}^{a_1}, ..., \mathcal{E}^{a_k}) \leq F_\cap(\mathcal{E}^{a_1}, ..., \mathcal{E}^{a_{k-1}})$.
    
\end{description}

Given a network redundancy function $F_\cap^\alpha$, one can also ask how a specific collection of networks $\alpha$ contributes to the overall utility of the joint network. To quantify this, we can define utility ``atoms'' that capture the contribution of $\alpha$ over and above the contribution of other elements lower in the lattice,
\begin{align}
    F_\partial^\alpha = F_\cap^\alpha - \sum_{\beta \prec \alpha} F_\partial^\beta ~ ,
    \label{eq:moebius}
\end{align}
\noindent where $\prec$ (and $\preceq$) correspond to the natural partial order in the antichain lattice~\cite{williams2010nonnegative}. This is equivalent to saying that $F_\partial^\alpha$ is the Moebius inversion of $F_\cap^\alpha$~\cite{williams2010nonnegative}, and can also be written as
\begin{align}
    F_{\cap}^\alpha = \sum_{\beta \preceq \alpha} F_{\partial}^\beta ~ ,
\end{align}
which together with the self-intersection property also implies that the sum of all atoms decomposes $F$ as expected,
\begin{align}
    F(\mathcal{E}) = F_\cap^{\{\{1,...,N\}\}} = \sum_{\alpha \in \mathcal{A}} F_\partial^\alpha ~ .
\end{align}
These atoms are the objects of interest in our analyses, and correspond to the redundant, unique, and synergistic contributions to global efficiency presented in the previous section.

In addition to the formalism above, we need one more ingredient to compute these atoms: a definition of the network redundancy function $F_\cap$. Although more definitions satisfying the properties above could certainly be possible, here we propose the following definition for its suitability and simplicity:
\begin{align}
    F_{\cap}^{\alpha} = \E{\min_{a \in \alpha} f(\Omega; \mathcal{E}^a)} ~ .
    \label{eq:net_red}
\end{align}
With $\Omega$ defined as above and $f$ being the efficiency (inverse of the shortest path length), it is easy to see that this definition satisfies the symmetry, self-intersection, and monotonicity axioms above. Furthermore, it is possible to prove that this definition of $F_{\cap}^{\alpha}$ is \textit{totally monotone} in $\mathcal{A}$, which guarantees that all network atoms $F_\partial^\alpha$ are non-negative (see proofs in the Appendix). With this definition, it is possible to evaluate Eq.~\eqref{eq:net_red} on all antichains, and then use Eq.~\eqref{eq:moebius} to compute all atoms, concluding the calculation.

Finally, it is worth noting that, in analogy with the previous section, we can we
can naturally define a redundancy function for node pairs $\omega$, i.e. $f_\cap^\alpha(\omega) = \min_{a \in \alpha} f(\omega; \mathcal{E}^a)$,
such that $F_\cap^\alpha = \E{f_\cap^\alpha (\omega)}$ (and similarly for $f_\partial^\alpha(\omega)$).
With this, we can directly generalise our previous definition of the dominant character of
a node pair $\omega$ as the set $\alpha$ such that $f_\partial^\alpha(\omega) > 0$ and $f_\partial^\beta(\omega) = 0 ~ \forall \beta \succ \alpha$. It is direct to show that there is always a unique $\alpha$ satisfying this condition for $N = 2$, and numerical experiments suggest this is also the case for $N > 2$ -- although, in the absence of a proof, we leave this as a conjecture for future work.

\subsection*{Data}

\subsubsection*{London transport networks}

The network of London public trail transport was obtained from \url{https://networks.skewed.de/net/london_transport}. Further details are available from the original publication by De Domenico \textit{et al.}~\citep{dedomenico2014pnas}. It is a multiplex network with 3 undirected edge types representing links within the three layers of London train stations: underground, overground and DLR. Here, we combined overground and DLR into a single network, and we then compared the respective contributions of underground versus overground+DLR (which we refer to as simply ``overground'').

\subsubsection*{Human structural connectomes from the\\Human Connectome Project}

We  used    diffusion  MRI  (dMRI)  data  from  the  $100$  unrelated  subjects (54 females and 46 males, mean age $= 29.1 \pm 3.7$ years)  of  the  HCP  $900$  subjects  data  release  \citep{vanessen2013neuroimg}. All  HCP  scanning  protocols  were  approved  by  the  local  Institutional  Review  Board  at  Washington  University  in  St.  Louis. The  diffusion-weighted imaging  (DWI)  acquisition  protocol  is  covered  in  detail  elsewhere~\citep{glasser2013neuroimg}. The  diffusion  MRI  scan  was  conducted  on  a  Siemens  3T  Skyra  scanner  using  a  2D  spin-echo  single-shot  multiband  EPI  sequence  with  a  multi-band  factor  of  3  and  monopolar  gradient  pulse.  The  spatial  resolution  was  $1.25$ mm  isotropic.  TR=$5500$ ms,  TE=$89.50$ms.  The  b-values  were  $1000$,  $2000$,  and  $3000$ s/mm$^2$.  The  total  number  of  diffusion  sampling  directions  was  $90$,  $90$,  and  $90$  for  each  of  the  shells  in  addition  to  $6$  b0  images.  We  used  the  version  of  the  data  made  available  in  DSI  Studio-compatible  format  at  \url{http://brain.labsolver.org/diffusion-mri-templates/hcp-842-hcp-1021}  \citep{yeh2018neuroimg}.

We adopted previously reported procedures to reconstruct the human connectome from DWI data. The  minimally-preprocessed  DWI  HCP  data~\citep{glasser2013neuroimg}  were  corrected  for  eddy  current  and  susceptibility  artifact.  DWI  data  were  then  reconstructed  using  q-space  diffeomorphic  reconstruction  (QSDR~\citep{yeh2011neuroimg}),  as  implemented  in  DSI  Studio  (\url{www.dsi-studio.labsolver.org}).  QSDR  is  a  model-free  method  that  calculates  the  orientation  distribution  of  the  density  of  diffusing  water  in  a  standard  space,  to  conserve  the  diffusible  spins  and  preserve  the  continuity  of  fiber  geometry  for  fiber  tracking.  QSDR  first  reconstructs  diffusion-weighted  images  in  native  space  and  computes  the  quantitative  anisotropy  (QA)  in  each  voxel.  These  QA  values  are  used  to  warp  the  brain  to  a  template  QA  volume  in  Montreal  Neurological  Institute  (MNI)  space  using  a  nonlinear  registration  algorithm  implemented  in  the  statistical  parametric  mapping  (SPM)  software.  A  diffusion  sampling  length  ratio  of  $2.5$  was  used,  and  the  output  resolution  was  1  mm. A  modified  FACT  algorithm  \citep{yeh2013pone}  was  then  used  to  perform  deterministic  fiber  tracking  on  the  reconstructed  data,  with  the  following  parameters~\citep{luppi2021networkneuro}: angular  cutoff  of  $55^{\circ}$,  step  size  of  $1.0$  mm,    minimum  length  of  $10$  mm,    maximum  length  of  $400$ mm,    spin  density  function  smoothing  of  $0.0$,    and  a  QA  threshold  determined  by  DWI  signal  in  the  cerebrospinal fluid.  Each  of  the  streamlines  generated  was  automatically  screened  for  its  termination  location.  A  white  matter  mask  was  created  by  applying  DSI  Studio's  default  anisotropy  threshold  ($0.6$  Otsu's  threshold)  to  the  spin distribution function’s anisotropy  values.  The  mask  was  used  to  eliminate  streamlines  with  premature  termination  in  the  white  matter  region.  Deterministic  fiber  tracking  was  performed  until  $1,000,000$  streamlines  were  reconstructed  for  each  individual.  

For each individual, their structural connectome was reconstructed by drawing an edge between each pair of regions from the Schaefer-200 cortical atlas~\citep{schaefer2018cercor} if there were white matter tracts connecting the corresponding brain regions end-to-end; edge weights were quantified as the number of streamlines connecting each pair of regions. 

A consensus matrix $A$ across subjects (consensus conenctome) was then obtained using the procedure of Wang and colleagues \citep{wang2019sciadv}, as follows: for each pair of regions $i$ and $j$, if more than half of subjects had non-zero connection between $i$ and $j$, $A_{ij}$ was set to the average across all subjects with non-zero connections between $i$ and $j$. Otherwise, $A_{ij}$ was set to zero.

\subsubsection*{Alternative structural connectome from Lausanne dataset}

A total of $N = 70$ healthy participants (25 females,
age $28.8 \pm 8.9$ years old) were scanned at the Lausanne
University Hospital in a 3T MRI scanner (Trio,
Siemens Medical, Germany) using a 32-channel head coil~\citep{griffa2019zenodo}. Informed
written consent was obtained for all participants in accordance with institutional guidelines and the protocol
was approved by the Ethics Committee of Clinical
Research of the Faculty of Biology and Medicine, University
of Lausanne, Switzerland. The protocol included: (1) a magnetization-prepared
rapid acquisition gradient echo (MPRAGE) sequence
sensitive to white/gray matter contrast (1 mm
in-plane resolution, 1.2 mm slice thickness); and (2) a diffusion
spectrum imaging (DSI) sequence ($128$ diffusion-weighted
volumes and a single b0 volume, maximum
$b$-value $8\,000$ s/mm$^2$, $2.2\times2.2\times3.0$ mm voxel size).

Structural connectomes were reconstructed for individual
participants using deterministic streamline tractography
and divided according to a subdivision of the Desikan-Killiany anatomical parcellation, with $234$ cortical and subcortical regions (chosen as the closest match for the 200-node Schaefer parcellation). White matter and grey matter were segmented
from the MPRAGE volumes using the FreeSurfer (version
5.0.0) open-source package, whereas DSI data preprocessing
was implemented with tools from the Connectome
Mapper open-source software, initiating
32 streamline propagations per diffusion direction for
each white matter voxel. Structural connectivity
was defined as streamline density between node pairs,
i.e., the number of streamlines between two regions normalized
by the mean length of the streamlines and the
mean surface area of the regions, following previous work with these data~\citep{vazquezrodriguez2019pnas, hagmann2008plosbiol}. 

\subsubsection*{Human functional connectomes}
We  used    resting-state functional  MRI  (rs-fMRI)  data  from  the  same 100  unrelated  subjects  of  the  HCP  900  subjects  data  release  \citep{vanessen2013neuroimg}. The following sequences were used: structural MRI: 3D MPRAGE T1-weighted, TR = 2,400 ms, TE = 2.14 ms, TI = 1,000 ms, flip angle = 8°, FOV = 224 × 224, voxel size = 0.7 mm isotropic. Two sessions of 15-min resting-state fMRI: gradient-echo EPI, TR = 720 ms, TE = 33.1 ms, flip angle = 52°, FOV = 208 by 180, voxel size = 2 mm isotropic. Here, we used functional data from only the first scanning session, in LR direction. 

\paragraph*{Functional MRI denoising.}
We used the minimally preprocessed fMRI data from the HCP, which includes bias field correction, functional realignment, motion correction, and spatial normalization to Montreal Neurological Institute (MNI-152) standard space with 2 mm isotropic resampling resolution~\citep{glasser2013neuroimg}. We also removed the first 10 volumes, to allow magnetization to reach steady state. Additional denoising steps were performed using the SPM12-based toolbox CONN (\url{http://www.nitrc.org/projects/conn}), version 17f~\citep{whitfield2012conn}. To reduce noise due to cardiac and motion artifacts, we applied the anatomical CompCor method of denoising the functional data. The anatomical CompCor method (also implemented within the CONN toolbox) involves regressing out of the functional data the following confounding effects: the first five principal components attributable to each individual’s white matter signal, and the first five components attributable to individual cerebrospinal fluid (CSF) signal; and six subject-specific realignment parameters (three translations and three rotations) as well as their first-order temporal derivatives~\citep{whitfield2012conn}. Linear detrending was also applied, and the subject-specific denoised BOLD signal time series were band-pass filtered to eliminate both low-frequency drift effects and high-frequency noise, thus retaining frequencies between $0.008$ and $0.09$ Hz.

\paragraph*{Functional network reconstruction.}
Functional connectivity (FC) networks were obtained for each subject by correlating the BOLD timeseries of each pair of regions in the Schaefer atlas. A group-average FC network was then obtained by averaging across all subjects. To ensure comparability of the structural and functional networks, the functional networks were each thresholded to have the same density as the corresponding structural network, a procedure termed ``structural density matching''~\citep{luppi2021networkneuro}. Both structural and functional networks were binarised before analysis.

\subsubsection*{Mammalian connectomes from diffusion MRI}

We used data available online (\url{https://doi.org/10.5281/zenodo.7143143}); below we provide the main information, with further detail available in the original publication~\citep{assaf2020conservation}. For consistency of reporting, where possible we use the same wording as recent publications using this dataset~\citep{assaf2020conservation, suarez2022connectomics}.

\paragraph*{Brain samples.}
The MaMI database includes a total of 220 ex vivo diffusion and T2- and T1-weighted brain scans of 125 different animal species. No animals were deliberately euthanized for this study. All brains were collected based on incidental death of animals in zoos in Israel or natural death collected abroad, and with the permission of the national park authority (approval no. 2012/38645) or its equivalent in the relevant countries. All scans were performed on excised and fixated tissue. Animals' brains were extracted within 24 hr of death and placed in formaldehyde (10\%) for a fixation period of a few days to a few weeks (depending on the brain size). Approximately 24 hr before the MRI scanning session, the brains were placed in phosphate-buffered saline for rehydration. Given the limited size of the bore, small brains were scanned using a 7T 30/70 BioSpec Avance Bruker system, whereas larger brains were scanned using a 3T Siemens Prisma system. To minimize image artefacts caused by magnet susceptibility effects, the brains were immersed in fluorinated oil (Flourinert, 3M) inside a plastic bag during the MRI scanning session.

\paragraph*{MRI acquisition}

A unified MRI protocol was implemented for all species. The protocol included high-resolution anatomical scans (T2- or T1-weighted MRI), which were used as an anatomical reference, and diffusion MR scans. Diffusion MRI data were acquired using high angular resolution diffusion imaging (HARDI), which consists of a series of diffusion-weighted, spin-echo, echo-planar-imaging images covering the whole brain, scanned in either 60 (in the 7T scanner) or 64 gradient directions (in the 3T scanner) with an additional three non-diffusion-weighted images (b0). The b value was $1000$ s/mm$^2$ in all scans. In the 7T scans, TR was longer than $12,000$ ms (depending on the number of slices), with TE of 20 ms. In the 3T scans, TR was $3,500$ ms, with a TE of 47 ms.

To linearly scale according to brain size the two-dimensional image pixel resolution (per slice), the size of the matrix remained constant across all species ($128 \times 96$). Due to differences in brain shape, the number of slices varied between 46 and 68. Likewise, the number of scan repetitions and the acquisition time were different for each species, depending on brain size and desired signal-to-noise ratio (SNR) levels. To keep SNR levels above 20, an acquisition time of 48 hr was used for small brains ($\sim$0.15 ml) and 25 min for large brains ($>1000$ ml). SNR was defined as the ratio of mean signal strength to the standard deviation of the noise (an area in the non-brain part of the image). Full details are provided in the original publication~\citep{assaf2020conservation}.

\paragraph*{Connectome reconstruction}

The ExploreDTI software was used for diffusion analysis and tractography~\citep{leemans2009exploredti}. The following steps were used to reconstruct fibre tracts. To reduce noise and smooth the data, anisotropic smoothing with a 3-pixel Gaussian kernel was applied. Motion, susceptibility, and eddy current distortions were corrected in the native space of the HARDI acquisition. A spherical harmonic deconvolution approach was used to generate fibre-orientation density functions per pixel, yielding multiple ($n \geq 1$) fibre orientations per voxel. Spherical harmonics of up to fourth order were used~\citep{tournier2004direct}. Whole-brain tractography was performed using a constrained spherical deconvolution (CSD) seed point threshold similar for all samples (0.2) and a step length half the pixel size. The end result of the tractography analysis is a list of streamlines starting and ending between pairs of voxels. Recent studies have shown that fibre tracking tends to present a bias where the vast majority of end points reside in the white matter~\citep{tournier2004direct}. To avoid this, the CSD tracking implemented here ensures that approximately 90\% of the end points reside in the cortical and subcortical grey matter.

\paragraph*{Network reconstruction}

Before the reconstruction of the structural networks, certain fibre tracts were removed from the final list of tracts. These include external projection fibres that pass through the cerebral peduncle, as well as cerebellar connections. Inner-hemispheric projections, such as the thalamic radiation, were included in the analysis. Brains were parcellated into 200 nodes using a k-means clustering algorithm. All the fibre end-point positions were used as input, and cluster assignment was done based on the similarity in connectivity profile between pairs of end points. Therefore, vertices with similar connectivity profile have a higher chance of grouping together. The clustering was performed twice, once for each hemisphere. Nodes were defined as the centre of mass of the resulting 200 clusters. Connectivity matrices were generated by counting the number of streamlines between any two nodes~\citep{assaf2020conservation}. The resulting connectivity matrices are hence sparse and weighted adjacency matrices. Matrices were binarised by setting connectivity values to 1 if the connection exists and 0 otherwise.

Even though the sizes of the regions differ across species, we opted for a uniform parcellation scheme (i.e. 200 nodes) for several reasons, in keeping with previous work~\citep{suarez2022connectomics}. First, to our knowledge, there is no MRI parcellation for the brains of the majority of the species studied here. Second, how brain regions correspond to one another across species (i.e. homologues) is still not completely understood for many regions and for many species. Third, comparing networks of different sizes introduces numerous analytical biases because most network measures trivially depend on size, making the comparison challenging. We therefore opted to implement a uniform parcellation scheme across species, allowing us to translate connectomes into a common reference feature space in which they can be compared. Note that this approach does not take into account species-specific regional delineations, nor does it capture homologies between nodes across species, which are still not completely understood.

\subsubsection*{Identification of long- and short-range connections}

For both human structural and functional networks, as well as mammalian structural networks, connection length was defined in terms of Euclidean distance between the centroids of the regions-of-interest constituting the endpoints of each edge. 

A short-range (resp., long-range) connectivity network was obtained as the shortest (resp., longest) $50\%$ of edges. Thus, for each starting network, we obtain a network of its $50\%$ shortest edges, and a network of its $50\%$ longest edges. This approach ensures that both networks have equal density and are therefore on equal terms in terms of their expected contribution to the composite network.

\subsection*{Network models}

Binary, undirected Erd\H{o}s-R\'enyi networks of different densities were created by randomly selecting a fraction $f$ of all possible bidirectional edges in the network, and setting them to unity, with $f$ ranging between $0.01$ and $1.0$ in increments of $0.01$. 

For the rewiring experiment, we initially generated a binary, undirected network of 200 nodes with lattice topology and $5\%$ density. A copy of this network was then generated and its edges progressively randomised in $1\%$ increments, using a standard rewiring procedure to preserve the degree of each node.

Small-world character was measured via the small-world propensity (SWP) index \citep{muldoon2016scirep}, which quantifies the deviation of the network’s empirically observed clustering coefficient, $C_{obs}$, and characteristic path length, $L_{obs}$, from equivalent lattice ($C_{latt}$, $L_{latt}$) and random ($C_{rand}$, $L_{rand}$) networks of equal number of nodes and edges: 

\begin{align}
\text{SWP} = 1 - \sqrt{(\Delta_{C}^2  + \Delta_{L}^2)/2} 
\end{align}

with 
\begin{align}
\Delta_{C}=(C_{latt}  - C_{obs})/(C_{latt}  - C_{rand}) \end{align}

and
\begin{align}
\Delta_{L}=(L_{obs}  - L_{rand})/(L_{latt}  - L_{rand})
\end{align}

such that SWP is bound between $0$ and $1$.

To disambiguate the role of connectome topology in shaping the contribution to shortest paths, we relied on network null models~\citep{vasa2022natrevneurosci}. Specifically, we adopted the well-known Maslov-Sneppen degree-preserving rewired network, whereby edges are swapped so as to randomise the topology while preserving the exact binary degree of each node (degree sequence)~\citep{maslov2002science}.

\subsection*{Statistical analysis}

For the London transport networks, statistical significance of the empirical results was assessed against a population of 1000 network null models, constructed as described above. For both human and animal brain networks, the null distribution was obtained by rewiring each individual network, and the empirical and null distributions were compared using permutation-based paired t-tests (with 10,000 permutations). Non-parametric tests were chosen to ensure robustness to outliers. Effect size was computed as Hedge's \textit{g}. Tables with full descriptive statistics for the results reported in the main text are provided in the Supplementary.

\section*{Acknowledgments}
The authors are grateful to Yaniv Assaf and the Strauss Center for Neuroimaging of Tel Aviv University for making the MaMI dataset available, and to Joe Lizier, Pradeep Banerjee, and members of the Network Neuroscience Lab for helpful discussion. This research was supported by the Visitor Program of the Max Planck Institute for Mathematics in the Sciences, Leipzig (Germany) [to AIL]. 
FR was supported by the Fellowship Programme of the Institute of Cultural and Creative Industries of the University
of Kent, and the DIEP visitors programme at the University of Amsterdam. 
JJ acknowledges support from Grant 1514 of the German Israeli Foundation.

\section*{Conflicts of interest}
None.

\noindent

\bibliography{main}


\appendix

\setcounter{equation}{0}
\setcounter{figure}{0}
\setcounter{table}{0}
\makeatletter
\renewcommand{\theequation}{S\arabic{equation}}
\renewcommand{\thefigure}{S\arabic{figure}}
\renewcommand{\thetable}{S\arabic{table}}

\section{Supplementary proofs}

This section contains supporting proofs for the properties of the efficiency-based Partial Network Decomposition (PND), presented in the Materials and Methods section of the main text. Our main goal is to prove that the proposed network redundancy function in Eq.~\eqref{eq:net_red} provides a non-negative PND of the network's average efficiency, such that $F_\partial^\alpha \geq 0, \forall \alpha$. We will do so by closely following the proofs for PID in Appendix D of Williams and Beer~\cite{williams2010nonnegative}, bearing in mind the differences between our network redundancy function $F_\cap^\alpha$ and their information redundancy function $I_\text{min}(S; \alpha)$. Note that some proofs in Ref.~\cite{williams2010nonnegative} do not depend on the properties of $I_\text{min}(S; \alpha)$, and rely only on the structure of the antichain lattice --- which is shared between PID and PND.

We begin by formally defining the efficiency $f$ of a pair of distinct nodes $\omega = (v_1, v_2) \in \mathcal{V} \times \mathcal{V}$ s.t. $v_1 \neq v_2$ as the minimum length of all paths between them in the set of edges $\mathcal{E}$:
\begin{align}
    f(\omega; \mathcal{E}) = \min_{\texttt{p} \in \texttt{P}(\omega; \mathcal{E})} |\texttt{p}|^{-1} ~ ,
    \label{eq:def_efficiency}
\end{align}
\noindent where $\texttt{P}(\omega; \mathcal{E})$ is the set of all paths between $v_1$ and $v_2$ in $\mathcal{E}$. For convenience, we repeat the definition of the redundancy function presented in the main text,
\begin{align}
    f_\cap^\alpha(\omega) = \min_{a\in\alpha} f(\omega; \mathcal{E}^a) ~ .
    \label{eq:def_red}
\end{align}
In the following, we may omit the dependence on $\omega$ for simplicity of notation. The following theorems (up to Theorem~\ref{th:F_nonneg}) hold for all $\omega$.

\begin{theorem}
    $f(\omega; \mathcal{E})$ is non-negative.
\end{theorem}
\begin{proof}
    Follows from the fact that path lengths are non-zero positive integers. Following standard convention, the shortest path length between disconnected nodes is taken to be positive infinity.
\end{proof}

\begin{lemma} \label{lemma:f_monotone}
    $f(\omega; \mathcal{E}^a)$ increases monotonically under subset inclusion.
\end{lemma}
\begin{proof}
    Consider $a, b \subseteq \{1,...,N\}$, with $a \subset b$. Recall that, by definition, $\mathcal{E}^a = \bigcup_{i=1}^k \mathcal{E}_{n_i}$ for any
    $a = \{n_1,..,n_k\} \subseteq \{1,\dots,N\}$. If $a \subset b$, then $\mathcal{E}^a \subseteq \mathcal{E}^b$. Since the set of paths $\texttt{P}(\omega; \mathcal{E})$ grows with the inclusion of more edges in $\mathcal{E}$, the minimum of any function on $\texttt{P}$ must decrease with said inclusion, and therefore its inverse must increase. Thus, if $a \subset b$ then $f(\omega; \mathcal{E}^a) \leq f(\omega; \mathcal{E}^b)$.
\end{proof}

\begin{theorem}\label{th:fcap_monotone}
    $f_\cap^\alpha$ increases monotonically in the redundancy lattice.
\end{theorem}
\begin{proof}
    This proof depends only on the structure of the redundancy lattice and Lemma~\ref{lemma:f_monotone}, both of which are shared between Ref.~\cite{williams2010nonnegative} and this work. Proof is exactly as in the original work, replacing $I_\text{min}(S; \alpha)$ by $f_\cap^\alpha$ and $I(S=s; a)$ by $f(\omega; \mathcal{E}^a)$.
\end{proof}

\begin{theorem}
    $f_\partial^\alpha$ can be written in closed form as
    \begin{align}
        f_\partial^\alpha = f_\cap^\alpha - \sum_{k=1}^{|\alpha^-|} (-1)^{k-1}\sum_{\substack{\mathcal{B}\subseteq \alpha^-\\|\mathcal{B}|=k}} f_\cap^{\bigwedge\mathcal{B}}~.
        \label{eq:f_meet}
    \end{align}
\end{theorem}
\begin{proof}
    This proof depends only on the structure of the redundancy lattice and Eq.~\eqref{eq:moebius} defining $f_\partial^\alpha$ as the Moebius inversion of $f_\cap^\alpha$, both of which are shared between Ref.~\cite{williams2010nonnegative} and this work. Proof is exactly as in the original work, replacing $I_\text{min}(S; \alpha)$ by $f_\cap^\alpha$ and $\Pi_{\text{\textbf{R}}}(S; \alpha)$ by $f_\partial^\alpha$.
\end{proof}

\begin{theorem}\label{th:f_maxmin}
    $f_\partial^\alpha$ can be written in closed form as
    \begin{align}
        f_\partial^\alpha = f_\cap^\alpha - \max_{\beta\in\alpha^-} \min_{b\in\beta} f(\omega; \mathcal{E}^b)~.
    \end{align}
\end{theorem}
\begin{proof}
    First we note that, for the redundancy lattice, the following statement holds:
    \begin{align}
        \alpha \wedge \beta = \underline{\alpha \cup \beta} ~ ,
        \label{eq:meet}
    \end{align}
    \noindent where $\underline{X}$ is the set of minimal elements of a poset $X$ (see Refs.~\cite{crampton2000two,williams2010nonnegative} for a proof). Combining Eqs.~\eqref{eq:def_red} and~\eqref{eq:f_meet} yield
    \begin{align}
        f_\partial^\alpha &= f_\cap^\alpha - \sum_{k=1}^{|\alpha^-|} (-1)^{k-1}\sum_{\substack{\mathcal{B}\subseteq \alpha^-\\|\mathcal{B}|=k}} \min_{b\in\bigwedge\mathcal{B}} f(\omega; \mathcal{E}^b)~, \\
        \intertext{and by Lemma~\ref{lemma:f_monotone} and Eq.\eqref{eq:meet},}
        &= f_\cap^\alpha - \sum_{k=1}^{|\alpha^-|} (-1)^{k-1}\sum_{\substack{\mathcal{B}\subseteq \alpha^-\\|\mathcal{B}|=k}} \min_{\beta\in\mathcal{B}} \min_{b\in\beta} f(\omega; \mathcal{E}^b) ~ . \\
        \intertext{Then, applying the maximum-minimums identity~\cite{williams2010nonnegative} we have}
        &= f_\cap^\alpha - \max_{\beta\in\alpha^-} \min_{b\in\beta} f(\omega; \mathcal{E}^b) ~ , \\
        \intertext{Or, equivalently, by applying Eq.~\eqref{eq:def_red} again,}
        &= f_\cap^\alpha - \max_{\beta\in\alpha^-} f_\cap^\beta ~ .
    \end{align}
\end{proof}

\begin{theorem}\label{th:f_nonneg}
    $f_\partial^\alpha$ is non-negative.
\end{theorem}
\begin{proof}
    If $\alpha$ is the infimum of $\mathcal{A}$, $f_\partial^\alpha = f_\cap^\alpha$, and $f_\partial^\alpha \geq 0$ follows from Theorem~\ref{th:f_nonneg}. If $\alpha$ is not the infimum of $\mathcal{A}$, $f_\partial^\alpha \geq 0$ follows directly from Theorems~\ref{th:fcap_monotone} and \ref{th:f_maxmin}.
\end{proof}

\begin{theorem}\label{th:F_nonneg}
    $F_\partial^\alpha$ is non-negative.
\end{theorem}
\begin{proof}
    Follows directly from Theorem~\ref{th:f_nonneg} and the fact that $F_\partial^\alpha = \E{f_\partial^\alpha}$ is a weighted sum of non-negative quantities with non-negative weights.
\end{proof}

\section{Supplementary figures and tables}

This section contains Figure~\ref{fig_Lausanne_S1000_replication}, replicating our main result in human structural brain networks in a different dataset; as well as Tables~\ref{table_SC_vs_Nulls} and~\ref{table_MAMI_vs_Nulls}, reporting detailed statistics of the comparison between our empirically observed results and those obtained in null network models.


\begin{figure*} 
\centering
\includegraphics[width=0.98\textwidth]{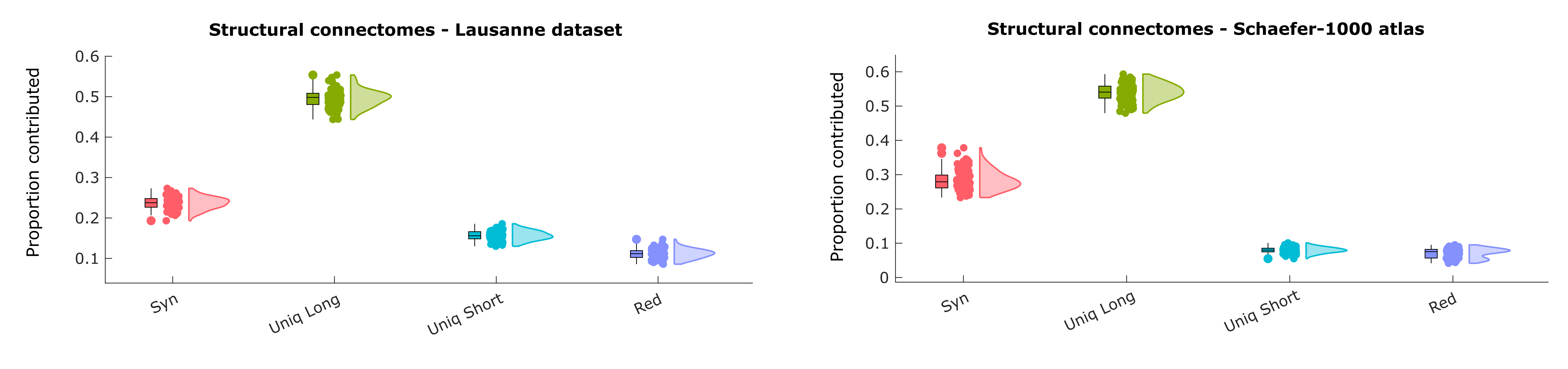}
\caption{{\bf Replication with alternative reconstruction of the human structural connectome}. Left: connectomes reconstructed from an independent dataset of Diffusion Spectrum Imaging data (N=$70$ subjects), using the Lausanne anatomical parcellation with 234 cortical and subcortical regions. Right: connectome reconstructed from an alternative sub-parcellation of the Schaefer functional atlas with $1000$ cortical nodes (N=$100$ subjects). Y-axis: proportion of shortest paths accounted for by each PID term. Box-plots indicate the median and inter-quartile range of the distribution. Each data-point is one subject }
\label{fig_Lausanne_S1000_replication}
\end{figure*}


\begin{table*}[]
\begin{tabular}{@{}lllllllllll@{}}
\toprule
                      & \multicolumn{1}{c}{\textbf{Real Mean}} & \multicolumn{1}{c}{\textbf{Real SD}} & \multicolumn{1}{c}{\textbf{Null Mean}} & \multicolumn{1}{c}{\textbf{Null SD}} & \multicolumn{1}{c}{\textbf{tStat}} & \multicolumn{1}{c}{\textbf{df}} & \multicolumn{1}{c}{\textbf{pVal}} & \multicolumn{1}{c}{\textbf{Hedges g}} & \multicolumn{1}{c}{\textbf{CI Lower}} & \multicolumn{1}{c}{\textbf{CI Upper}} \\ \midrule
\textbf{Synergy}      & 1.84E-01                               & 1.94E-02                             & 1.32E-01                               & 2.05E-02                             & 34.85                              & 99                              & p \textless\ 0.001                 & 2.57                                  & 2.31                                       & 2.92                                      \\ 
\textbf{Unique Long}  & 4.97E-01                               & 2.15E-02                             & 2.60E-01                               & 8.24E-03                             & 115.65                             & 99                              & p \textless\ 0.001                 & 14.49                                 & 12.84                                      & 16.93                                     \\
\textbf{Unique Short} & 1.58E-01                               & 1.10E-02                             & 2.51E-01                               & 8.23E-03                             & -60.48                             & 99                              & p \textless\ 0.001                 & -9.43                                 & -10.71                                     & -8.46                                     \\
\textbf{Redundancy}   & 1.61E-01                               & 1.72E-02                             & 3.57E-01                               & 3.09E-02                             & -114.75                            & 99                              & p \textless\ 0.001                 & -7.79                                 & -8.95                                      & -6.99                                     \\ \hline
\end{tabular}
\caption{Results for the statistical comparison between human structural connectivity networks, and corresponding degree-preserving rewired nulls.}
\label{table_SC_vs_Nulls}
\end{table*}

\begin{table*}[]
\begin{tabular}{@{}lllllllllll@{}}
\toprule
                      & \multicolumn{1}{c}{\textbf{Real Mean}} & \multicolumn{1}{c}{\textbf{Real SD}} & \multicolumn{1}{c}{\textbf{Null Mean}} & \multicolumn{1}{c}{\textbf{Null SD}} & \multicolumn{1}{c}{\textbf{tStat}} & \multicolumn{1}{c}{\textbf{df}} & \multicolumn{1}{c}{\textbf{pVal}} & \multicolumn{1}{c}{\textbf{Hedges g}} & \multicolumn{1}{c}{\textbf{CI Lower}} & \multicolumn{1}{c}{\textbf{CI Upper}} \\ \midrule
\textbf{Synergy}      & 9.19E-02                               & 4.46E-02                             & 4.94E-02                               & 4.03E-02                             & 48.01                              & 219                             & p \textless 0.001                 & 1                                     & 0.85                                       & 1.22                                      \\ 
\textbf{Unique Long}  & 3.43E-01                               & 3.38E-02                             & 2.64E-01                               & 1.98E-02                             & 28.15                              & 219                             & p \textless 0.001                 & 2.87                                  & 2.6                                        & 3.19                                      \\
\textbf{Unique Short} & 2.19E-01                               & 2.31E-02                             & 1.98E-01                               & 1.93E-02                             & 13.85                              & 219                             & p \textless 0.001                 & 1.02                                  & 0.85                                       & 1.2                                       \\
\textbf{Redundancy}   & 3.45E-01                               & 4.94E-02                             & 4.89E-01                               & 7.12E-02                             & -50.02                             & 219                             & p \textless 0.001                 & -2.34                                 & -2.73                                      & -2.04                                     \\ \hline 
\end{tabular}
\caption{Results for the statistical comparison between mammalian structural connectivity networks, and corresponding degree-preserving rewired nulls.}
\label{table_MAMI_vs_Nulls}
\end{table*}




\end{document}